%% file: main.tex
\documentclass[submission,copyright,creativecommons]{eptcs}
\providecommand{\event}{LTT2026} 

\usepackage[T1]{fontenc}        
\usepackage[utf8]{inputenc}
\usepackage{bold-extra} 

\usepackage{float}
\floatstyle{boxed}
\restylefloat{figure}

\usepackage{xspace}
\usepackage{graphicx}
\usepackage[dvipsnames]{xcolor}
\definecolor{dkgreen}{rgb}{0,0.6,0}
\definecolor{ltblue}{rgb}{0,0.4,0.4}
\definecolor{dkviolet}{rgb}{0.3,0,0.5}
\definecolor{dkblue}{rgb}{0.0,0.1,0.6}
\usepackage{soul}

\usepackage{bussproofs}

\usepackage{stmaryrd}
\usepackage{amsmath}
\usepackage{amsfonts}
\usepackage{amsthm}
\theoremstyle{plain}
\newtheorem{theorem}{Theorem}

\newtheorem{corollary}{Corollary}
\theoremstyle{definition}
\newtheorem{definition}{Definition}
\newtheorem{example}{Example}
\theoremstyle{remark}
\newtheorem{remark}{Remark}

\usepackage{listings}
\newcommand{\listingsttfamily}{\fontfamily{SourceCodePro-TLF}\small}
\input{coqlst.tex}

\input{clst.tex}

\newenvironment{dedication}
{
    \itshape             
    \raggedleft          
}
{
    \par 
}

\input{macros.tex}

\title{Reversible Computation with\\ Stacks and ``Reversible Management of Failures''}
\author{
Matteo Palazzo \qquad\qquad Luca Roversi
\institute{Dipartimento di Informatica\\ Università di Torino -- Italy}
\email{\quad matteo.palazzo@unito.it \quad\qquad luca.roversi@unito.it}
}
\def\titlerunning{Reversible Computation, Stacks, Management of Failures}
\def\authorrunning{M. Palazzo \& L. Roversi}

\usepackage{hyperref}
\hypersetup{
  bookmarksnumbered,
  pdftitle    = {\titlerunning},
  pdfauthor   = {\authorrunning},
  pdfsubject  = {EPTCS},               
}

\begin{document}
\maketitle

\input{abstract.tex}
\begin{dedication}
    Dedicated to Stefano Berardi\\
    on the occasion of his 64th birthday.
\end{dedication}
\input{introduction.tex}
\input{related-work.tex}
\input{closing-dedication.tex}

\input{structure-paper.tex}
\input{score-syntax.tex}

\input{score-naive-semantic.tex}
\input{score-assert-semantic.tex}
\input{score-semantic.tex}

\input{conclusion.tex}

\bibliographystyle{eptcs}
\bibliography{bibliography.bib}


\end{document}

%% file: coqlst.tex
\lstdefinelanguage{Coq}{
	mathescape=true,
	texcl=false,
	morekeywords=[1]{Section, Module, End, Require, Import, Export,
		Variable, Variables, Parameter, Parameters, Axiom, Hypothesis,
		Hypotheses, Notation, Local, Tactic, Reserved, Scope, Open, Close,
		Bind, Delimit, Definition, Let, Ltac, Fixpoint, CoFixpoint, Add,
		Morphism, Relation, Implicit, Arguments, Unset, Contextual,
		Strict, Prenex, Implicits, Inductive, CoInductive, Record,
		Structure, Canonical, Coercion, Context, Class, Global, Instance,
		Program, Infix, Theorem, Lemma, Corollary, Proposition, Fact,
		Remark, Example, Proof, Goal, Save, Qed, Defined, Hint, Resolve,
		Rewrite, View, Search, Show, Print, Printing, All, Eval, Check,
		Projections, inside, outside, Def},
	morekeywords=[2]{forall, exists, exists2, fun, fix, cofix, struct,
		match, with, end, as, in, return, let, if, is, then, else, for, of,
		nosimpl, when},
	morekeywords=[3]{Type, Prop, Set, true, false, option, Z, list, nat},
	morekeywords=[4]{pose, set, move, case, elim, apply, clear, hnf,
		intro, intros, generalize, rename, pattern, after, destruct,
		induction, using, refine, inversion, injection, rewrite, congr,
		unlock, compute, ring, field, fourier, replace, fold, unfold,
		change, cutrewrite, simpl, have, suff, wlog, suffices, without,
		loss, nat_norm, assert, cut, trivial, revert, bool_congr, nat_congr,
		symmetry, transitivity, auto, split, left, right, autorewrite, reflexivity},
	morekeywords=[5]{by, done, exact, tauto, romega, omega,
		assumption, solve, contradiction, discriminate},
	morekeywords=[6]{do, last, first, try, idtac, repeat},
	morekeywords=[7]{push, pop},
	morecomment=[s]{(*}{*)},
	showstringspaces=false,
	morestring=[b]",
	morestring=[d]’,
	tabsize=3,
	extendedchars=false,
	sensitive=true,
	breaklines=false,
	basicstyle=\footnotesize\listingsttfamily,
	captionpos=b,
	columns=[l]flexible,
	identifierstyle={\footnotesize\listingsttfamily\color{black}},
	keywordstyle=[1]{\bfseries\footnotesize\listingsttfamily\color{dkviolet}},
	keywordstyle=[2]{\bfseries\footnotesize\listingsttfamily\color{dkgreen}},
	keywordstyle=[3]{\bfseries\footnotesize\listingsttfamily\color{ltblue}},
	keywordstyle=[4]{\bfseries\footnotesize\listingsttfamily\color{blue}},
	keywordstyle=[5]{\footnotesize\listingsttfamily\color{red}},
	keywordstyle=[7]{\bfseries\footnotesize\listingsttfamily\color{BrickRed}},
	stringstyle=\footnotesize\listingsttfamily,
	commentstyle={\footnotesize\listingsttfamily\color{dkgreen}},
	literate=
	{\\forall}{{\color{dkgreen}{$\forall\;$}}}1
	{\\exists}{{$\exists\;$}}1
	{<-}{{$\leftarrow\;$}}1
	{=>}{{$\Rightarrow\;$}}1
	{==}{{\code{==}\;}}1
	{==>}{{\code{==>}\;}}1
	{->}{{$\rightarrow\;$}}1
	{<->}{{$\leftrightarrow\;$}}1
	{<==}{{$\leq\;$}}1
	{\#}{{$^\star$}}1
	{\\o}{{$\circ\;$}}1
	{\@}{{$\cdot$}}1
	{\/\\}{{$\wedge\;$}}1
	{\\\/}{{$\vee\;$}}1
	{++}{{\code{++}}}1
	{~}{{$\sim$}}1
	{\@\@}{{$@$}}1
	{\\mapsto}{{$\mapsto\;$}}1
	{\\hline}{{\rule{\linewidth}{0.5pt}}}1
    {\ }{{\ }}1
}[keywords,comments,strings]

%% file: clst.tex
\lstdefinestyle{c}{
    language=C,
    morekeywords={assert, true, false, then, fi, PUSH, POP, INC, DEC, FOR, SKIP, let, in},  
	backgroundcolor=\color{backcolour},
    basicstyle=\footnotesize\listingsttfamily,
    breakatwhitespace=false,         
    breaklines=true,                 
    belowcaptionskip=.5\baselineskip, 
    captionpos=b,                    
    commentstyle=\color{purple!40!black},
    deletekeywords={...},            
    escapeinside={\%*}{*)},          
    extendedchars=true,              
    identifierstyle=\color{blue},
    firstnumber=0,                
    stepnumber=1,                
    frame=trbl,
    keepspaces=true,                 
    keywordstyle=\bfseries\color{green!35!black},
    numbers=none,                    
    numbersep=10pt,                   
    numberstyle=\tiny\color{mygray}, 
    rulecolor=\color{black},         
    showspaces=false,                
    showstringspaces=false,          
    showtabs=false,                  
    stepnumber=2,                    
    stringstyle=\color{orange},
    tabsize=2,	                   
    title=\lstname,                   
    xleftmargin=\parindent,
    }

%% file: macros.tex
\newcommand{\MH}[1]{\textcolor{red}{#1}}

\newcommand{\LH}[1]{\textcolor{blue}{\textbf{#1}}}
\newcommand{\PIF}{{\normalfont\textsf{PIF-reversibilization}}\xspace}
\newcommand{\TIF}{{\normalfont\textsf{TIF-reversibilization}}\xspace}
\newcommand{\nsemantics}{{\normalfont\textsf{N-semantics}}\xspace}
\newcommand{\asemantics}{{\normalfont\textsf{A-semantics}}\xspace}
\newcommand{\rsemantics}{{\normalfont\textsf{R-semantics}}\xspace}
\newcommand{\StrRev}{{\normalfont\textsf{Strongly-reversibile}}\xspace}
\newcommand{\WeaRev}{{\normalfont\textsf{Weakly-reversibile}}\xspace}
\newcommand{\SCORE}{{\normalfont\textsf{S-CORE}}\xspace}
\newcommand{\SRL}{{\normalfont\textsf{SRL}}\xspace}
\newcommand{\JANUS}{{\normalfont\textsf{Janus}}\xspace}
\newcommand{\RCORE}{{\normalfont\textsf{R-CORE}}\xspace}
\newcommand{\C}{{\normalfont\textsf{C}}\xspace}
\newcommand{\Coq}{\href{https://coq.inria.fr}{\normalfont\textsf{Coq/Rocq}}\xspace}
\newcommand{\lstiC}[1]{\mbox{\normalfont\lstinline[style=c]|#1|}}

\newcommand{\vCoq}[1]{\mbox{\normalfont\lstinline[language=Coq]|#1|}}
\newcommand{\invSC}[1]{\lstiC{-#1}}
\newcommand{\NN}{\mathbb{N}}
\newcommand{\ZZ}{\mathbb{Z}}
\newcommand{\StackZZ}{\mathcal{S}_{\ZZ}}
\newcommand{\nil}{[\,]}
\newcommand{\cons}{\!::\!}
\newcommand{\hd}{\operatorname{hd}}
\newcommand{\tl}{\operatorname{tl}}


%% file: abstract.tex
\begin{abstract}
This work examines approaches to making computational models reversible. Broadly speaking, transforming a computational model into a reversible one, i.e. reversibilizing it, means extending its operational semantics conservatively in a way that each term of the model is interpretable as a bijection. 
We recall that the most common strategy to reversibilize a computational model yields operational semantics that halts computations whenever a computational state cannot be uniquely determined from its successor state, thereby allowing terms to be interpreted as \emph{partial} bijective functions.
We are interested in reversible computational models whose terms can be interpreted as \emph{total} bijective functions. This is essential for studying aspects of computational complexity related to reversible computational models.
We introduce \SCORE, a language designed for manipulating variables and stacks. 
Notably, common reversibilization strategies naturally lead to interpreting the functions for stack manipulation as partial bijections.
According to our interests, we demonstrate how to interpret \SCORE in a state space where, using a proof-assistant, we certify that stack operations are \emph{total} bijections. It follows that all \SCORE terms can be interpreted as \emph{total} bijections.
\end{abstract}

%% file: introduction.tex
\section{Introduction}
\label{section:Introduction}
Landauer's Principle \cite{landauer1961irreversibility} states that irreversible (i.e., classical) computation inherently dissipates energy, as it typically erases bits of information. This insight led to the exploration of reversible computational models and paradigmatic programming languages as alternatives to classical ones because inherently free from information erasure and unnecessary energy dissipation.

Fredkin and Toffoli gates \cite{Toffoli1980,FredkinToffoli1982} provided the foundational basis for designing reversible algorithms. Since then, reversible computational models have been introduced to advance our understanding of computation, characterized by the following core principles:
\begin{itemize}
    \item In terms of computable functions, reversible computation can only process injective functions \cite{Axelsen2011WhatDR}.
    \item Regarding operational semantics, every step in the forward direction from input to output can always be traced back uniquely recovering the input from the output.
\end{itemize}

\noindent
For example, reversible computation implemented by reversible programming languages can offer practical advantages:
\begin{itemize}
    \item It simplifies debugging and testing by enabling programs to run backward, allowing for more effective error tracing and behavioral analysis \cite{LaneseG:rc24}.
    \item It naturally provides primitives that facilitate efficient checkpointing and rollback mechanisms-critical for long-running applications. These mechanisms enable programs to revert to previous states without losing progress, ensuring smooth recovery from errors and interruptions \cite{perumalla2013introduction}.
\end{itemize}

\paragraph{``Reversibilization''.}
A common approach to defining the operational meaning of ``reversible computation'' is to develop techniques for \emph{``reversibilizing''} the operational semantics of classical programming models or languages, whether deterministic or non-deterministic. We propose two distinct descriptions that can be used to clarify how ``reversibilization'' can be obtained in a deterministic setting. 

\paragraph{``Reversibilize to Partial Injective Functions (\PIF)''.}
This form of reversibility is realized by allowing the operational semantics of a reversible computational model to abort the interpretation of a term when the current state does not produce a new state from which we can recover the current one unambiguously. 
Consequently, terms are interpreted as partial injective mappings: (i) partial because they are not defined for every input; (ii) injective because, when no abort occurs, 
both the next and the previous computational states can be
uniquely determined independently from the state in which the computation is.

The reversible programming language \JANUS \cite{LUTZ:JANUS86} is an instance of \PIF. Its operational semantics extends that one of a Turing-complete fragment of the \C programming language by incorporating \lstiC{assert} statements that enforce properties specific to computational states\footnote{\emph{State of a computation}: the values of relevant variables at a given point in the computation.}. 

Verifying computational states via \lstiC{assert} is crucial for deriving the ``backward'' interpretation of \JANUS programs from output to input. Any computation -- whether executing in the ``forward'' or ``backward'' direction -- that enters a state failing to meet an \lstiC{assert} condition is aborted.

\begin{example}
We show why \JANUS \lstiC{if-then-else} is reversible.
Recall that \JANUS \lstiC{if-then-else} has the form
``\lstiC{if b then P else Q fi a}'' where \lstiC{b} is the usual guard and \lstiC{a} is a boolean expression.
The ``forward'' interpretation of ``\lstiC{if b then P else Q fi a}'' works as expected: if \lstiC{b} evaluates to \lstiC{true}, then \lstiC{P} is executed; otherwise \lstiC{Q} is executed.
Let $\sigma$ be the state obtained after either \lstiC{P} or \lstiC{Q} is executed. At this point, the \JANUS semantics checks \lstiC{assert a} in the state $\sigma$.
It is the programmer’s task to define \lstiC{a} so that if $\sigma$ comes from \lstiC{P}, then \lstiC{assert a} holds, and if $\sigma$ comes from \lstiC{Q}, then \lstiC{not(assert a)} holds. In all other cases, the ``forward'' interpretation of ``\lstiC{if b then P else Q fi a}'' aborts.

In the ``backward'' direction, the roles of \lstiC{b} and \lstiC{a} are swapped. This means that we interpret ``\lstiC{if a then P else Q fi b}'' instead of ``\lstiC{if b then P else Q fi a}'', where \lstiC{a} is now the guard and \lstiC{b} is checked by \lstiC{assert}.
Here, \lstiC{a} decides which branch to follow, and \lstiC{assert b} must evaluate to \lstiC{true} or \lstiC{false}, depending on the branch just taken, to avoid \emph{abort}.

The key idea is that \lstiC{b} and \lstiC{a} together ensure that
\lstiC{P} is executed ``forward'' if and only if it is also executed ``backward'', and
\lstiC{Q} is executed ``forward'' if and only if it is also executed ``backward''.
\end{example}

\paragraph{``Reversibilize to Total Injective Functions (\TIF)''.}
This form of reversibilization can be achieved by refining the representation of the computational states that the terms in a classical computational model manipulate. The objective is for these terms to be interpreted as \emph{total} injective functions. While this limit expressivity, reversible programs always correspond to total functions, ensuring they never fail.

\SRL \cite{MATOS:TCS15} is an instance of \TIF. It is a reversible programming language in which variables, unlike in many other models, range over $\mathbb{Z}$, instead of only $\NN$, to assure that the increment or the decrement of any value is always reversible, i.e.\! that one operation can always be undone by the other.
The operational semantics of \SRL always terminates and does not require \lstiC{assert} on computational states, a property that justifies our terminology. As a result, \SRL is less expressive than \JANUS. Still, it can encode every \emph{primitive recursive function} \cite{MalettoR:JLAMP24}, as long as an \lstiC{if-then-else} structure is included \cite{MatosPR:RC20}.
The ``backward'' execution of any program \lstiC{P} in \SRL is the same as the ``forward'' execution of its inverse \lstiC{-(P)}, where \lstiC{-(P)} is \emph{syntactically generated} from \lstiC{P}.

\begin{example}
A \SRL program is:
\begin{align*}
\lstiC{INC x; FOR x (DEC y); FOR x \{INC y\}; DEC x}
\enspace .
\end{align*}
\noindent
For any initial value $u$ in \lstiC{x}, ``\lstiC{INC x}'' increments it. Consequently, ``\lstiC{FOR x \{DEC y\}}'' decreases the value $v$ in \lstiC{y} for $u+1$ times. Interpreting ``\lstiC{INC x; FOR x \{DEC y\}}'' sets \lstiC{y} to $v-(u+1)$.
The second part ``\lstiC{FOR x \{INC y\}; DEC x}'' \emph{is the inverse} ``\invSC{(INC x; FOR x \{DEC y\})}'' of the first part ``\lstiC{INC x; FOR x \{DEC y\}}'': increasing $u+1$ times \lstiC{y} and decreasing \lstiC{x} recovers the initial state with $u$ in \lstiC{x} and $v$ in \lstiC{y}.
\end{example}

\paragraph{Motivations.}
Between \PIF and \TIF, we focus on the latter. The reason is that we find \PIF insufficient for fully characterizing reversible computations -- for example, when addressing questions about the computational complexity of reversible computational models.

\paragraph{Contributions.}
We study the idea that, with a small amount of extra information, we can move from cases where ``the instruction \lstiC{I} in a reversible computational model $\mathcal{M}$ sometimes needs to \emph{abort} because it is not reversible in every state'' to cases where ``the instruction \lstiC{I} in $\mathcal{M}$ is always reversible because the computational states in which \lstiC{I} runs are \emph{properly refined}''.

Our view is that \lstiC{assert} is not strictly required to define expressive reversible computational models, covering both algorithms and the functions they compute.

We make this view concrete by extending \SRL to \SCORE. Compared to \SRL, \SCORE is more algorithmically expressive because it has variables that hold $\mathbb{Z}$-valued stacks, with \lstiC{PUSH} and \lstiC{POP} acting on them. In particular, the operational semantics of \SCORE uses refined states in which \lstiC{PUSH} and \lstiC{POP}, absent in \SRL, are \emph{never-aborting} mutual inverses. That is, \lstiC{PUSH} and \lstiC{POP} form a total isomorphism. This is achieved by designing the \emph{internal} stack representation to keep the following analogy:
\begin{quote}
``\lstiC{PUSH} on a stack \lstiC{s} corresponds to the increment \lstiC{INC} of some $\mathbb{Z}$-valued variable \lstiC{x}, just as \lstiC{POP} on \lstiC{s} corresponds to the decrement \lstiC{DEC} of \lstiC{x}.'' 
\end{quote}
\noindent
In particular, with the proof assistant \href{https://coq.inria.fr/}{\textsf{Coq/Rocq}}:
(i) we define two functions \vCoq{push} and \vCoq{pop}, which act as the operational semantics of \lstiC{PUSH} and \lstiC{POP};
(ii) we formally prove that \vCoq{push} and \vCoq{pop} are mutual inverses within a suitable three-dimensional domain of interpretation.
The related code is at \cite{PlazzoRoversi24RevPushPopGit}.

\begin{example}
A \SCORE program is:
\begin{align*}
\lstiC{FOR x \{POP s\}; FOR x \{PUSH s\}}
\enspace .
\end{align*}
\noindent
Let $3$ be the initial value of \lstiC{x}, and assume that \lstiC{s} contains the stack \lstiC{[2,1]}. In \SCORE:
\begin{itemize}
\item
interpreting ``\lstiC{FOR x \{POP s\}}'' empties \lstiC{s} beyond its bottom element because the iteration tries to remove more items than those initially present;
\item
interpreting ``\lstiC{FOR x \{PUSH s\}}'', i.e. the inverse ``\invSC{(FOR x \{POP s\})}'' of ``\lstiC{FOR x \{POP s\}}'', immediately after ``\lstiC{FOR x \{POP s\}}'', restores the original stack \lstiC{[2,1]} in \lstiC{s}.
\end{itemize}
\end{example}

\vspace{\baselineskip}
We conclude by noting that we use a step-by-step approach to build \vCoq{push} and \vCoq{pop}.
The reason is partly pedagogical, given the potentially broad readership of this work.
In particular, we introduce three different operational semantics of \SCORE, called \nsemantics, \asemantics, and \rsemantics.
The aim is to gradually guide the reader to see the structure of the states where \vCoq{push} and \vCoq{pop} are interpreted, ultimately removing the need for \lstiC{assert}, in line with the \TIF perspective.

%% file: related-work.tex
\paragraph{Related work.}
We see the operational semantics of \lstiC{PUSH} and \lstiC{POP}, which ensure their mutual invertibility regardless of the initial stack content, as an instance of what Glück and Yokoyama term the ``\emph{Injectivisation process}'' in \cite{GLUCK2023113429}. 
Their focus is on the ``injectivisation of a function $f$,'' a process that \emph{extends} the co-domain of $f$. The output $f(x)$ is paired with an additional value $g(x)$, provided by a suitable function $g$, enabling the unique recovery of the input $x$ from the pair $(f(x), g(x))$. 

Drawing an analogy to injectivization, the structure used to interpret the stack in \SCORE is enhanced to ensure that \lstiC{PUSH} and \lstiC{POP} function as perfect inverses in all cases.

Furthermore, we see \SCORE as a paradigmatic imperative reversible language, analogous to \RCORE \cite{Glck2017AMR}. Both frameworks feature a minimal set of built-in primitives for reversible algorithm expression. The key distinction is that \RCORE aligns with \PIF, as its operational semantics utilizes \lstiC{assert} to guarantee reversibility when handling the ``tree'' data type. In contrast, \SCORE conforms to \TIF, adopting an \lstiC{assert}-free operational semantics.

%% file: closing-dedication.tex
\paragraph{This work as a tribute to Stefano.}
We decided to focus this sincere tribute to Stefano because we believe it can span both past and present aspects of Stefano's scientific interests. 
Recently, Stefano and the second author supervised a Doctoral thesis whose subject was partly on reversible computational models, resulting in \cite{DBLP:journals/tcs/BarileBR25}. 
In the past, Stefano worked also on dead-code elimination, which contributed to the development of both the proof assistant \Coq and the project \href{https://en.wikipedia.org/wiki/CompCert}{\textsf{CompCert}}, clearly relevant for supporting the importance of proof-assistants in building reliable and dependable software.

%% file: structure-paper.tex
\paragraph{Structure of this document.}
Section~\ref{section:S-CORE} introduces \SCORE.
Section~\ref{section:A naive semantic ldots irreversible} introduces the naive semantics \nsemantics.
Section~\ref{section:S-CORE with assertions} introduces \asemantics according to \PIF.
Section~\ref{section:Assertion free S-CORE} defines \rsemantics according to \TIF.

%% file: score-syntax.tex

\section{\textsf{S-CORE}}
\label{section:S-CORE}
We introduce \textsf{S-CORE} as an extension of \textsf{SRL} \cite{Matos:TCS03}, a reversible computational model that exclusively represents total injective functions. 
Two key features distinguish \SCORE from \SRL:
\begin{enumerate}
    \item \SCORE includes the primitives \lstiC{PUSH} and \lstiC{POP} for stack management;
    \item Each variable \lstiC{x} in \SCORE is represented as a triple consisting of:
    \begin{itemize}
        \item The first component, storing the current value of \lstiC{x};
        \item The second component, maintaining the history of values that \lstiC{x} has assumed;
        \item The third component, enabling \lstiC{PUSH} to function as a successor operation and \lstiC{POP} as a predecessor operation, ensuring their mutual reversibility.
    \end{itemize}
\end{enumerate}
\noindent
Before proceeding, let us recall the main features of \SRL.

\subsection{\SRL in a nutshell.}
\label{subsection:SRL in a nutshell}
\textsf{SRL} is a reversible imperative computational model 
in which every term \lstiC{P} syntactically determines its inverse \lstiC{-P}. We will illustrate this aspect technically when introducing \SCORE, the subject of this work.
The variables of \SRL range over $\ZZ$. The instruction ``\lstiC{INC x}'' increments the value of \lstiC{x} by one, whereas ``\lstiC{DEC x}'' decrements it by one. Moreover,
given that $v$ represents the value of \lstiC{x}, the loop construct ``\lstiC{FOR x P}'' iterates the body \lstiC{P} exactly $|v|$ times, following these rules:
\begin{itemize}
    \item If $v > 0$, then \lstiC{P} is \emph{iterated forward}, meaning each iteration executes \lstiC{P}.
    \item If $v < 0$, then the body is \emph{iterated backward}, meaning each iteration executes the \emph{inverse} \lstiC{-P} of \lstiC{P}.
    \item If $v = 0$, then \lstiC{P} is skipped entirely.
\end{itemize}

\subsection{\SCORE syntax}
\label{subsection:SCORE: syntax and first notions}
\begin{definition}[\textsf{S-CORE} syntax]
\label{definition:S-CORE syntax}
Let $V$ be a set of variable names \lstiC{x}, \lstiC{y}, \lstiC{z}, \ldots \enspace. Let $\lstiC{x} \in V$. \SCORE terms belong to the language generated by the grammar:
\begin{align*}
	\lstiC{P} & ::= \lstiC{SKIP} \mid \lstiC{INC x} \mid \lstiC{DEC x}
	\mid \lstiC{PUSH x} \mid \lstiC{POP x} \mid \lstiC{P;P}  \mid
    \lstiC{FOR x \{P\}}
    \enspace ,
\end{align*}
under the proviso that the \emph{leading variable} \lstiC{x} in ``\lstiC{FOR x \{P\}}'' \emph{cannot} occur in \lstiC{P}, that is \lstiC{P} never contains neither ``\lstiC{INC x}'' nor ``\lstiC{DEC x}''. Writing $\lstiC{P} \in \SCORE$ means that \lstiC{P} is an \SCORE term.
\end{definition}

\begin{remark}
\SRL contains all and only the \SCORE terms which are free of ``\lstiC{PUSH x}'' and ``\lstiC{POP x}'', for every variable \lstiC{x}.
\end{remark}

\begin{definition}[Inverse of a term in \SCORE]
\label{definition:Inverse -P of any SCORE term P}
The function \invSC{(.)}$:\SCORE$ $\to \SCORE$ inverts every $\lstiC{P}\in\SCORE$ by structural induction on \lstiC{P}:
\begin{alignat*}{3}
\invSC{(INC x)} &= \lstiC{DEC x}
&\quad
\invSC{SKIP}    &= \lstiC{SKIP}
&\quad
\invSC{(DEC x)} &= \lstiC{INC x}
\\
\invSC{(PUSH x)} & = \lstiC{POP x}
&& &
\invSC{(POP x)}  & = \lstiC{PUSH x}
\\
\invSC{(P;Q)} & = \invSC{Q}\lstiC{;}\invSC{P}
&& &
\invSC{(FOR x \{P\})} & = \lstiC{FOR x \{}\invSC{P}\lstiC{\}}
\enspace .
\end{alignat*}
\end{definition}

\begin{remark}
It should be obvious to see that the function \invSC{(.)} is \emph{self-dual}, i.e. $\lstiC{-(-(P))} = \lstiC{P}$, for every \SCORE program \lstiC{P}. 
\end{remark}

\subsection{Our plan towards the final operational semantics \rsemantics of \SCORE}

Clearly, in a classical computational model, \lstiC{PUSH} and \lstiC{POP} are not a pair of bijections because it is meaningless to \lstiC{POP} elements from an empty stack. We need to extend the structure of the domain in which we interpret \lstiC{PUSH} and \lstiC{POP} to eventually obtain that both ``\lstiC{-(PUSH)} = \lstiC{POP}'' and ``\lstiC{-(POP)} = \lstiC{PUSH}'' also hold at the operational semantics level.

We think it is useful to follow a step-wise extension of the domain to operationally interpret \lstiC{PUSH} and \lstiC{POP}, to show a concrete example of how to inject non reversible functions into reversible ones.

Our step-wise approach is realized by defining three big-step operational semantics -- \nsemantics, \asemantics, and \rsemantics\xspace -- and the corresponding domain of interpretation they rely on, following two general guidelines:
\begin{enumerate}
    \item Each semantics is based on its own notion of \emph{state} $\sigma$, mapping every variable to an appropriate current value.
    \item Each semantics is formalized as a relation with judgments of the form
    $\langle \lstiC{P}, \sigma \rangle \Downarrow \tau$
    which we interpret as 
        ``For every $\lstiC{P} \in \SCORE$, and for every state $\sigma$ and $\tau$, executing \lstiC{P} from $\sigma$ results in $\tau$.'' 
\end{enumerate}\noindent

%% file: score-naive-semantic.tex

\input{S-core-naive-sos.tex}

\section{A classical non-reversible semantics for \SCORE to start with}
\label{section:A naive semantic ldots irreversible}

We here introduce \nsemantics  for \SCORE, its `\textsf{N}' standing for `naive', word that emphasizes that the choices to define \nsemantics are simplistic, even though they provide a starting point for the other semantics introduced later.

\vspace{\baselineskip}\noindent
The set $\StackZZ$ of all and only the stacks containing values of $\ZZ$ is inductively given as:
\begin{itemize}
\item
$[\,] \in \StackZZ$, and
\item
$h \cons t \in \StackZZ$ if $h \in \mathbb{Z}$ and $t \in \StackZZ$
\end{itemize}
and is associated with the following head and tail functions:
\begin{align}
\label{align: hd and tl base}
\hd(\nil)     &= 0  &\tl(\nil)     &= \nil \\
\label{align: hd and tl ind}
\hd(h\cons t) &= h  &\tl(h \cons t) &= t
\enspace ,
\end{align}
\noindent
which are both \emph{total}, coherently with our interest about \TIF.

\begin{definition}[\nsemantics]
\label{definition:States of nsemantics}
Fig.~\ref{fig:scoreSemanticnaive} shows the rules of \nsemantics which operates on states $\sigma : V \to \ZZ \times \StackZZ$ such that, if $\sigma(\lstiC{x}) = (v, s)$ then:
\begin{itemize}
\item
$v$ is the value currently hold by \lstiC{x} in $\sigma$;
\item
$s$ is the stack currently hold by \lstiC{x} in $\sigma$, containing all and only the values previously stored in \lstiC{x} immediately before every interpretation of ``\lstiC{PUSH x}'' which resets \lstiC{x} to $0$.
\end{itemize}
\noindent
States are ranged over by $\Sigma$.
As usual, $\sigma[\lstiC{x} \to (v', s')]$ is $\sigma$ where \lstiC{x} maps to $(v', s')$, for some $v', s'$, and every \lstiC{y} other than \lstiC{x} maps to $\sigma(\lstiC{y})$.
\end{definition}
\noindent
The rule \textsc{push} interprets ``\lstiC{PUSH x}'' by pushing $v$ on top of $s$ and replacing $v$ by $0$, for any $\sigma$ and $x$ such that $\sigma(x) = (v, s)$. 
The rule \textsc{pop} interprets ``\lstiC{POP x}'' by popping the top $h$ of $s$ and using $h$ in place of $v$.

\subsection{\nsemantics always terminates}
The reason why, for every \lstiC{P}, and state $\sigma$, the interpretation of \lstiC{P} in $\sigma$ by means of \nsemantics always terminates in some unique state $\tau$ is almost straightforward.
Every iteration of ``\lstiC{FOR x \{P\}}'' generates the interpretation of a finite sequence ``\lstiC{P;...;P}'', regardless of whether ``\lstiC{FOR x \{P\}}'' is interpreted using \textsc{for-fwd} or \textsc{for-bwd}, as both unfold into a tree of multiple \textsc{step} occurrences with a single occurrence of \textsc{base} as its leaf.
Both ``\lstiC{INC x}'' and ``\lstiC{DEC x}'' rely on $\ZZ$ being infinite in both directions, meaning that there is no inherent limit to the number of times they can be applied.
The same observation holds for ``\lstiC{PUSH x}'' and ``\lstiC{POP x}''. Even when $\sigma(\lstiC{x}) = (v, \nil)$, ``\lstiC{POP x}'' can be applied indefinitely because, based on \eqref{align: hd and tl base}, $\nil$ is a fixed point of its interpretation $\tl(\cdot)$.

\subsection{\nsemantics is not reversible}
\label{subsection:nsemantics is not reversible}

This is because, for every \lstiC{x}, at least one state $\sigma$ exists such that \LH{($\langle \lstiC{POP x;} \invSC{(POP x)}, \sigma \rangle \not\Downarrow \sigma$)}, according to the rules in Fig.~\ref{fig:scoreSemanticnaive}.

For example, let $\sigma$ be such that $\sigma(\lstiC{x}) = (5,[2])$. 
Then, $\langle \lstiC{POP x}, \sigma \rangle \Downarrow \sigma[\lstiC{x} \to (2, \nil)]$. As $\lstiC{PUSH x} =$  $\invSC{(POP x)}$, we have $\langle \invSC{(POP x)}, \sigma[\lstiC{x} \to (2, \nil)] \rangle \Downarrow \sigma[\lstiC{x} \to (0,[2])]$, i.e. $\langle \lstiC{POP x;}\invSC{(POP x)}, \sigma \rangle \Downarrow \sigma[\lstiC{x} \to (0,[2])]$, but $\sigma[x \to (0,[2])] \neq \sigma$, which means that the sequential composition of ``\lstiC{POP x}'', followed by its inverse ``\lstiC{PUSH x}'', starting from $\sigma$, does not lead back to $\sigma$.

The problem is that ``\lstiC{POP x}'' overwrites the current value of \lstiC{x}. Therefore, when trying to reverse the operation by applying ``\lstiC{PUSH x}'', the value to which \lstiC{x} must be reset is unknown. The \nsemantics will simply always set \lstiC{x} to the default value $0$. This is certainly not surprising, remarking that if we want be coherent with \TIF perspective, the operational semantics of \SCORE must be more sophisticated than \nsemantics.

%% file: S-core-naive-sos.tex
\begin{figure}
	\centering
	\begin{tabular}{c}
		\bottomAlignProof
		\AxiomC{$ $}
		\RightLabel{\scriptsize \textsc{skip}}
		\UnaryInfC{$\langle \lstiC{SKIP}, \sigma \rangle \Downarrow \sigma$}
		\DisplayProof
		\qquad
		\bottomAlignProof
		\AxiomC{$\sigma(\lstiC{x}) = (v, s)$}
		\RightLabel{\scriptsize \textsc{inc}}
		\UnaryInfC{$\langle \lstiC{INC x}, \sigma \rangle \Downarrow \sigma[\lstiC{x} \to (v+1,s)]$}
		\DisplayProof
		\\ \\
		\bottomAlignProof
		\AxiomC{$\sigma(\lstiC{x}) = (v, s)$}
		\RightLabel{\scriptsize \textsc{dec}}
		\UnaryInfC{$\langle \lstiC{DEC x}, \sigma \rangle \Downarrow \sigma[\lstiC{x} \to (v-1,s)]$}
		\DisplayProof
		\qquad
		\bottomAlignProof
		\AxiomC{$\sigma(\lstiC{x}) = (v, s)$}
		\RightLabel{\scriptsize \textsc{push}}
		\UnaryInfC{$\langle \lstiC{PUSH x}, \sigma \rangle \Downarrow \sigma[\lstiC{x} \to (0,v::s)]$}
		\DisplayProof
		\\ \\
		\bottomAlignProof
		\AxiomC{$\sigma(\lstiC{x}) = (v,
        s)$}
		\RightLabel{\scriptsize \textsc{pop}}
		\UnaryInfC{$\langle \lstiC{POP x}, \sigma \rangle \Downarrow \sigma[\lstiC{x} \to (\hd(s),\tl(s))]$}
		\DisplayProof
		\qquad
		\bottomAlignProof
		\AxiomC{$\langle \lstiC{P} , \sigma \rangle \Downarrow \nu $}
		\AxiomC{$\langle \lstiC{Q} , \nu \rangle \Downarrow \tau $}
		\RightLabel{\scriptsize \textsc{seq}}
		\BinaryInfC{$\langle \lstiC{P;Q}, \sigma \rangle \Downarrow \tau$}
		\DisplayProof
        \\ \\
		\bottomAlignProof
		\AxiomC{$\sigma(\lstiC{x}) = (v,s)$}
		\AxiomC{$v \geq 0$}
		\AxiomC{$\langle \lstiC{P}, \sigma \rangle \Downarrow^v \tau$}
		\RightLabel{\scriptsize \textsc{for-fwd}}
		\TrinaryInfC{$\langle \lstiC{FOR x \{P\}}, \sigma \rangle \Downarrow \tau$}
		\DisplayProof
		\\ \\
		\bottomAlignProof
		\AxiomC{$\sigma(\lstiC{x}) = (v,s)$}
		\AxiomC{$v < 0$}
		\AxiomC{$
         \langle \lstiC{-(P)}, \sigma \rangle \Downarrow^{-v} \tau$}
		\RightLabel{\scriptsize \textsc{for-bwd}}
		\TrinaryInfC{$\langle \lstiC{FOR x \{P\}}, \sigma \rangle \Downarrow \tau$}
		\DisplayProof
		\\ \\
		\bottomAlignProof
		\AxiomC{$ $}
		\RightLabel{\scriptsize \textsc{base}}
		\UnaryInfC{$\langle \lstiC{P}, \sigma \rangle \Downarrow^0 \sigma$}
		\DisplayProof
		\qquad
		\bottomAlignProof
		\AxiomC{$\langle \lstiC{P}, \sigma \rangle \Downarrow \nu $}
		\AxiomC{$\langle \lstiC{P}, \nu \rangle \Downarrow^{n} \tau$}
		\RightLabel{\scriptsize \textsc{step}}
		\BinaryInfC{$\langle \lstiC{P}, \sigma \rangle \Downarrow^{n+1} \tau$}
		\DisplayProof
	\end{tabular}
	\caption{\nsemantics of \SCORE, the naive one.}
	\label{fig:scoreSemanticnaive}
\end{figure}

%% file: score-assert-semantic.tex

\section{A second reversible semantics for \SCORE\xspace -- this time with \lstiC{assert}}
\label{section:S-CORE with assertions}

Let us recall from the introduction that we identify \PIF as the standard approach for obtaining a reversible computational model from a non-reversible one. Instances of \PIF abound in the literature .
Broadly speaking, \PIF consists of introducing expressions like $\lstiC{assert(}\Phi_{\lstiC{P}}\lstiC{)}$ into an operational semantics, with 
\lstiC{P} a term of a computational model and $\Phi_{\lstiC{P}}$ a predicate. Applying $\lstiC{assert(}\Phi_{\lstiC{P}}\lstiC{)}$ to a computational state $\sigma$ means that:
\begin{itemize}
\item
\lstiC{P} can be interpreted in $\sigma$ the goal being to transform $\sigma$ into a new state $\tau$ if and only if $\Phi_{\lstiC{P}}$ holds on $\sigma$;
\item
the properties enforced by $\Phi_{\lstiC{P}}$ on $\sigma$ are sufficient to assure that the inverse \invSC{P}, when interpreted in $\tau$, correctly recovers $\sigma$.
\end{itemize}

\paragraph{The key aspect.} 
To obtain \asemantics according to \PIF, we must characterize the properties of states on which ``\lstiC{POP x}'' can be safely applied to produce a new state that ``\lstiC{PUSH x}'' can subsequently transform back into the original state from which ``\lstiC{POP x}'' was initially applied.


\begin{definition}[\asemantics]
\label{definition:asemantics}
The operational semantics \asemantics for \SCORE contains all and only the rules in {\normalfont Fig.~\ref{fig:scoreSemanticnaive}} of \nsemantics but \textsc{pop}, replaced by:
\begin{center}
\bottomAlignProof
\AxiomC{$\sigma(\lstiC{x}) = (v, s)$}
\AxiomC{$\lstiC{assert(}v\ \lstiC{=}\ 0 \lstiC{)} $}
\AxiomC{$\lstiC{assert(}s\ \lstiC{=}\ h \cons t \lstiC{)} $}
\RightLabel{\scriptsize \textsc{assert-pop} \enspace .}
\TrinaryInfC{$\langle \lstiC{POP x}, \sigma \rangle \Downarrow \sigma[\lstiC{x} \to (h,t)]$}
\DisplayProof 
\end{center}
\end{definition}
\noindent
The behavior of the new rule \textsc{assert-pop} can be illustrated by the following pseudo-code:
\begin{align}
\nonumber
&\lstiC{let}\ \sigma(x) = (v,s)\ \lstiC{in}
\\
\nonumber
&\lstiC{if assert(}v\ \lstiC{==}\ 0 \lstiC{) then}
\\
\label{align: code explaining assert-pop}
&\lstiC{ \{if assert(}s\ \lstiC{==}\ h \cons t \lstiC{) then \{ } ``\textrm{update}\ \sigma\ \textrm{to}\ \sigma[\lstiC{x} \to (h,t)]\textrm{''} \lstiC{ \} else abort\}}
\\
\nonumber
&\lstiC{else abort}
\enspace .
\end{align}
\noindent
The pseudo-code starts by evaluating \lstiC{x} in $\sigma$, getting some pair $(v, s)$, analogously to the rule \textsc{pop} it replaces in Fig.~\ref{fig:scoreSemanticnaive}.
The second step is to \emph{verify} that $\lstiC{assert(}v\ \lstiC{==}\ 0 \lstiC{)}$ holds in the state $\sigma$. If not, the whole computation \emph{aborts}. Otherwise, the last step \emph{verifies} $\lstiC{assert(}s\ \lstiC{==}\ h \cons t \lstiC{)}$ in $\sigma$, that is it checks that some $h$ and $t$ exist to match the stack $s$. The interpretation of ``\lstiC{POP x}'' keeps proceeding only in the positive case. Otherwise the computation \emph{aborts}.
 
As it should be evident from~\eqref{align: code explaining assert-pop} that ``\lstiC{P;}\invSC{P}'' does not necessarily lead to a computational state,
we think it is natural to introduce the following notion:

\begin{definition}[A \WeaRev semantics for \SCORE]
\label{definition:WeaRev semantics for SCORE}
A big-step semantic $\Downarrow\, \subseteq \SCORE \times \Sigma \times \Sigma$ for \SCORE is \WeaRev if and only if
$\langle \lstiC{P}, \sigma \rangle \Downarrow \tau
\iff
\langle \invSC{P}, \tau \rangle \Downarrow \sigma$, for every $\lstiC{P} \in \SCORE$ and $\sigma, \tau \in \Sigma$.
\end{definition}
\noindent
and to prove that \asemantics is \WeaRev according to it.

\begin{remark} 
Definition~\ref{definition:WeaRev semantics for SCORE} captures the idea that \invSC{P} can be applied to a state $\tau$ produced by \lstiC{P}, provided than \lstiC{P} does not abort, and vice versa.
This is weaker than requiring ``$\langle \lstiC{(P;}\invSC{P)}, \sigma \rangle \Downarrow \sigma$, for every $\sigma$'', which would imply that neither \lstiC{P} nor \invSC{P} can abort. In the coming sections, such a stronger requirement will be identified as ``\StrRev big-step semantics''.
\end{remark}

\begin{theorem}
\label{theorem: asemantics is eaRev}
The \asemantics in {\normalfont Definition~\ref{definition:asemantics}} is \WeaRev.
\end{theorem}

\begin{proof}
The proof can be carried out by structural induction on any term \lstiC{P} of \SCORE. We focus on proving the most relevant base case 
``$\langle \lstiC{POP x}, \sigma \rangle \Downarrow \tau
   \iff
   \langle \invSC{(POP x)}, \tau \rangle \Downarrow \sigma$'',
for every $\sigma$, and $\tau$.

\begin{itemize}
\item 
Let us start from the ``only if'' implication, assuming $\langle \lstiC{POP x}, \sigma \rangle \Downarrow \tau$.
The assumption says that, by the rule \textsc{assert-pop} in Definition~\ref{definition:asemantics}, some $v, s$ exist such that:
	\begin{align*}
		\sigma(\lstiC{x}) &= (0, v\cons s)
		&&
		\tau = \sigma[\lstiC{x} \to (v,s)]
		\enspace .
	\end{align*}
	\noindent
Interpreting ``$\invSC{(POP x)}$'' in $\sigma[\lstiC{x} \to (v,s)]$ means to interpret ``$\lstiC{PUSH x}$'' in $\sigma[\lstiC{x} \to (v,s)]$ by the rule \textsc{push} in Fig.~\ref{fig:scoreSemanticnaive}. Once both \textsc{assert-pop} and \textsc{push} have been applied, the final state is:
	\begin{align*}
		\underbrace{(\sigma[\lstiC{x} \to (v,s)])}_{\tau}[\lstiC{x} \to (0,v\cons s)]
		& = \sigma[\lstiC{x} \to (0,v\cons s)] = \sigma
		\enspace .
	\end{align*}

\item 
Let us turn to the ``if'' implication, assuming $\langle \invSC{(POP x)}, \tau \rangle \Downarrow \sigma$.
Since ``\invSC{(POP x)} = \lstiC{PUSH x}'', the assumption says that, by the rule \textsc{push} in Fig.~\ref{fig:scoreSemanticnaive}, some $v, s$ exist such that:
	\begin{align*}
		\tau(\lstiC{x}) & = (v, s)
            &&
		\sigma = \tau[\lstiC{x} \to (0,v\cons s)]
		\enspace .
	\end{align*}
	\noindent
Interpreting ``$\lstiC{(POP x)}$'' in $\tau[\lstiC{x} \to (0,v\cons s)]$ yields:
	\begin{align*}
		\underbrace{(\tau[\lstiC{x} \to (0,v\cons s)])}_{\sigma}[\lstiC{x} \to (v,s)]
		& = \tau[\lstiC{x} \to (v, s)] = \tau
		\enspace .
	\end{align*}
\end{itemize}
\end{proof}

\subsection{Some concluding observations on \asemantics}
\label{subsection:concerningAssertions}

Let $\llparenthesis \lstiC{P} \rrparenthesis$ be the function from $\Sigma$ to $\Sigma$ that \asemantics associates with a term \lstiC{P}, that is $\llparenthesis \lstiC{P} \rrparenthesis(\sigma) = \tau$ if and only if $\langle \lstiC{P}, \sigma \rangle \Downarrow \tau $, according to Definition~\ref{definition:asemantics}.

\begin{enumerate}
\item
\label{enumerate:interpretationA}
As \asemantics may abort, the interpretation $\llparenthesis \lstiC{P} \rrparenthesis$ of \lstiC{P} is necessarily a partial function $\Sigma \rightharpoonup\Sigma$ because the state ``abort'' does not naturally belong to $\Sigma$: once in ``abort'' we have no clue at which the state generated it, position shared at least by \cite{paoliniTYPES2015}.

\item
\label{enumerate:interpretationB}
Of course, defining $\Sigma_{\bot}$ as $\Sigma \cup \{\bot\}$, with $\bot$ distinguished from any state in $\Sigma$ to interpret ``abort'', we can transform $\llparenthesis \lstiC{P} \rrparenthesis$ into a total function in $\Sigma \rightarrow\Sigma_\bot$ such that $\llparenthesis \lstiC{P} \rrparenthesis(\sigma) = \bot$ if and only if $\langle \lstiC{P}, \sigma \not\Downarrow \tau \rangle$.
However, we emphasize that $\llparenthesis \lstiC{P} \rrparenthesis$ would not be \emph{globally injective}. Instead, it would be \emph{injective} only on a sub-domain of $\Sigma$, i.e., on all states that never lead to abort. We recall that such an interpretation has been explored in \cite{PalazzoRoversi:Forest}.
\end{enumerate}

%% file: score-semantic.tex

\section{A final, \lstiC{assert}-free, hence total, semantics of \textsf{S-CORE}}
\label{section:Assertion free S-CORE}

Let us recall from the introduction that we identify \TIF as the set of all approaches which, starting from a non-reversible model, have as their goal the definition of reversible operational semantics that, by design, can never abort because of an inability to reconstruct the computational state preceding the current one.

\subsection{\rsemantics}
\rsemantics is an \lstiC{assert}-free instance of \TIF. 
It's main aspect is to map every variable \lstiC{x} to a refinement of sates as triples $(v, s, c) \in \ZZ \times \StackZZ \times \NN$ whose last element $c$ acts as a counter. 
States-as-triples eliminates the need to abort computations that involve \lstiC{PUSH} and \lstiC{POP}, ensuring reversibility.

As we aim to formalize part of \rsemantics using the proof-assistant \Coq, we adopt \Coq notation ``\vCoq{ident -> Z*(list Z)*nat}'' to formalize the states-as-triples notation:
\begin{enumerate}
\item  
``\vCoq{ident}'' stands for the set of variable names $V$ in Definition~\ref{definition:S-CORE syntax};  
\item
``\vCoq{Z}'' stands for $\ZZ$;  
\item
``\vCoq{list Z}'' corresponds to the set $\StackZZ$ of all and only the stacks containing values of $\ZZ$;  
\item
``\vCoq{nat}'' represents $\NN$, and ``\vCoq{*}'' denotes the Cartesian product between two types.  
\end{enumerate}
\noindent
\begin{definition}[States as triples]
\label{definition:States as triples}
A state is a map $\sigma : \vCoq{ident -> Z*(list Z)*}$ \vCoq{nat}, where \vCoq{ident} is taken as an alias of the standard type \vCoq{string}.
\end{definition}

\subsection{The never aborting bijections \vCoq{push} and \vCoq{pop}}
\label{subsection:The main reason of nat occurring in Z*(list Z)*nat}

The structure of states in Definition~\ref{definition:States as triples} allows us to define the functions \vCoq{push} and \vCoq{pop} in \Coq, verify their properties \cite{PlazzoRoversi24RevPushPopGit}, and use them for interpreting \lstiC{PUSH} and \lstiC{POP}.

\begin{definition}[The bijections on stacks]
\label{definition:Functions push and pop}
Figure~\ref{fig:definitions of push and pop} defines \vCoq{push} and \vCoq{pop} as inductive functions of type \vCoq{Z*(list Z)*nat -> Z*(list Z)*nat} which proceed by pattern matching on their argument.
\end{definition}
\noindent
If $\sigma$ is a state such that $\sigma(\lstiC{x})$ yields \vCoq{(v,s,c):Z*(list Z)*nat}, then both \vCoq{push} and \vCoq{pop} assure that \vCoq{c} satisfies the \emph{invariant}:
\begin{quote}
``\vCoq{c} is greater than \vCoq{0} if and only if some ``\lstiC{POP x}'', interpreted by \vCoq{pop}, has been applied to states associated to \lstiC{x} in which \vCoq{s} is the empty list \vCoq{[]}, or \vCoq{v} is not $0$.''
\end{quote}
\noindent
The intuition is that the value of \vCoq{c} counts the number of \emph{illegal} \vCoq{pop} applications which generate a state in which we dub any variable \lstiC{x} as ``\emph{broken}'' as soon as its counter \vCoq{c} is greater than \vCoq{0}. 
The effects of illegal \vCoq{pop} applications can be ``undone'' by applying a corresponding number of \vCoq{push}.

\begin{figure}[t]
\begin{lstlisting}[language=Coq]
        Definition push (trpl: Z*(list Z)*nat) : Z*(list Z)*nat :=
        match trpl with
        | (v,    t,   O) => (0, v::t,   O)  (* 1st clause *)
        | (0, v::t, S c) => (0, v::t, S c)  (* 2nd clause *)
        | (u,    s, S c) => (u,    s,   c)  (* 3rd clause *)
        end.

        Definition pop (trpl: Z*(list Z)*nat) : Z*(list Z)*nat :=
        match trpl with
        | (0, v::t,   O) => (v,    t,   O)  (* 1st clause *)
        | (0, v::t, S c) => (0, v::t, S c)  (* 2nd clause *)
        | (u,    s,   c) => (u,    s, S c)  (* 3rd clause *)
        end.
\end{lstlisting}
\noindent
\begin{minipage}{.95\textwidth}
in which:
\begin{itemize}
\item
``\vCoq{match trpl with ...}'' matches the value of ``\vCoq{trpl: Z*(list Z)*nat}'' with the patterns specified to the left of the symbol ``\vCoq{=>}''. The pattern to the right of ``\vCoq{=>}'' is the corresponding result of \vCoq{push} or \vCoq{pop}.
\item
The symbol ``\vCoq{O}'', which is not ``\vCoq{0}'', represents the \emph{zero} of natural numbers in \vCoq{nat} in unary notation. The term ``\vCoq{S c}'' is the ``successor'' of the natural number in ``\vCoq{c: nat}''.
\item
``\vCoq{v::t}'' the list with head \vCoq{v} and tail \vCoq{t}.
\end{itemize}
\end{minipage}
\caption{The bijections \vCoq{push} and \vCoq{pop}.}
\label{fig:definitions of push and pop}
\end{figure}


\subsection{The bijections \vCoq{push} and \vCoq{pop}}
\label{subsection:Detailing out the behavior of pop}

Even though Fig.~\ref{fig:definitions of push and pop} first lists \vCoq{push} and then \vCoq{pop}, we started from designing \vCoq{pop} because the most problematic one, concerning \TIF.
Then we defined \vCoq{push} to be the inverse of \vCoq{pop}. 

\paragraph{The bijection \vCoq{pop}.}
Let us assume that $\sigma(\lstiC{x})$ yields \vCoq{(v,s,c)}, for some fixed \lstiC{x} and let us assume we apply \vCoq{pop} to \vCoq{(v,s,c)}.

\begin{enumerate}
\item
The \vCoq{(* 1st clause *)} applies when we can effectively pop a value out of the stack \vCoq{s}.
This is possible because we do not incur in any information loss because, simultaneously we have that:
\begin{itemize}
\item
the current value of \lstiC{x} is ``\vCoq{0:Z}'';
\item
\lstiC{x} is not \emph{broken} because \vCoq{c} holds ``\vCoq{O:nat}'';
\item
the stack \vCoq{s} contains at least one value.
\end{itemize}
The new triple associated to \lstiC{x} by \vCoq{pop} has the top of \vCoq{s} set to the new value of \lstiC{x}.

\item
The \vCoq{(* 2nd clause *)} applies when \lstiC{x} is broken  because \vCoq{c} \emph{is not} \vCoq{O:nat}, even though the stack \vCoq{s} is not empty, and we could pop its top to replace the current value \vCoq{0:Z} of \vCoq{v}, at least in principle.

This is the point where it is worth underlining the role of the counter \vCoq{c}.

Let us pretend to drop it, working with pairs \vCoq{(v,s)}, instead than triples, as we were using \nsemantics or \asemantics.

Both \vCoq{(* 2nd clause *)} and \vCoq{(* 3rd clause *)} of \vCoq{pop} in Definition~\ref{definition:Functions push and pop} would become:
\begin{lstlisting}[language=Coq]
        ...
        | (0, v::t) => (v, t)  (* new 2nd clause *)
        | (u,    s) => (u, s)  (* new 3rd clause *)
\end{lstlisting}
\noindent
They would prevent \vCoq{pop} to be injective, hence to be reversible, at least in the following case:
\begin{align}
\label{align: pop counterexample}
& \vCoq{pop(0,1::2::[])}
= \vCoq{(1,2::[])}
= \vCoq{pop(1,2::[])}
\enspace .
\end{align}

If we switch again to triples, then \eqref{align: pop counterexample} is no more a counterexample to the injectivity of \vCoq{pop} because:
\begin{align*}
& \vCoq{pop(0,1::2::[],0)}
= \vCoq{(1,2::[],0)}
\neq \vCoq{(1,2::[],S c)}
= \vCoq{pop(1,2::[],c)}
\enspace ,
\end{align*}
where, clearly, ``$\vCoq{S c} \neq \vCoq{0}$''.

In general, the counter \vCoq{c} serves as a \emph{flag}. When \vCoq{c} exceeds ``\vCoq{O:nat}'', it distinguishes states where applying \vCoq{pop} is meaningless from those where \vCoq{pop} behaves as expected, removing an element from the stack.

\item
The \vCoq{(* 3rd clause *)} applies:
\begin{itemize}
\item either when trying to pop an element from an empty stack,
\item or when trying to pop an element from a stack when the value \vCoq{v} associated to \vCoq{x} is not \vCoq{0}.
\end{itemize}
In both cases the counter \vCoq{c} is increased, possibly \emph{breaking} (if not already) the invoveld variable.
\end{enumerate}

\paragraph{The bijection \vCoq{push}.}
It is defined by three clauses, each designed to be the inverse of the corresponding one in \vCoq{pop}.

\begin{itemize}
\item
The \vCoq{(* 1st clause *)} applies to a non broken variable, i.e. when the counter \vCoq{c} holds \vCoq{O:nat}. It pushes \vCoq{v} on top of the stack, yielding \vCoq{v::t} and setting the value of \vCoq{v} to \vCoq{0}.

\item
The remaining clauses manage the \emph{broken} configurations, characterized by values of \vCoq{c} which counts former illegal applications of \vCoq{pop}. In those cases \vCoq{push} decreases the value in \vCoq{c}, eventually setting it to ``\vCoq{O:nat}'', i.e, possibly \emph{repairing} a broken variable.
\end{itemize}

\input{S-core-definitive-sos.tex}

\subsection{Properties of \vCoq{push} and \vCoq{pop}}
\label{subsection:Properties of push and pop}

We can finally prove that \vCoq{push} and \vCoq{pop} are each other inverse:
{%
\normalfont
\begin{lstlisting}[language=Coq, label=lemma:popPushInverse]
        Theorem pop_inv_push_inv_pop : forall p : Z*(list Z)*nat,
            (pop (push p) = p) /\ (push (pop p) = p).
\end{lstlisting}
}
\noindent
One can find the proof of the here above ``\vCoq{Theorem pop_inv_push_inv_pop}'' in \cite{PlazzoRoversi24RevPushPopGit}, based on the proof of its two components ``\vCoq{pop (push p) = p}'' and ``\vCoq{push (pop p) = p}'' which, in their turn, are proved as a case-by-case analysis of how \vCoq{pop} and \vCoq{push} behave when one is used just after the other.

\vspace{\baselineskip}\noindent
With \vCoq{pop} and \vCoq{push} available we can finally give:
\begin{definition}[\rsemantics of \SCORE]
\label{definition:rsemantics of SCORE}
Fig.~\ref{fig:scoreSemanticDefinitive} collects all and only its rules.
\end{definition}

\subsection{Strong reversibility of \SCORE}
Let us introduce the following concept:
\begin{definition}[\StrRev semantics of \SCORE]
\label{definition:StrRev semantics for SCORE}
Let $\Sigma$ be the set of all states $\sigma$. A big-step semantics
$\Downarrow \subseteq \SCORE \times \Sigma \times \Sigma$
is \StrRev if and only if
$\langle \lstiC{(P;}\invSC{P)}, \sigma \rangle \Downarrow \sigma$,
for every $\lstiC{P} \in \SCORE$ and $\sigma \in \Sigma$. 
\end{definition}
\noindent
The above ``\vCoq{Theorem pop_inv_push_inv_pop}'' implies:
\begin{corollary}[\SCORE is \StrRev]
\label{corollary:score_is_reversible}
Let \lstiC{P} be a term of \SCORE, a and $\sigma \in \Sigma$.
Then:
\begin{align*}
&
\langle \lstiC{(P;}\invSC{P)}, \sigma \rangle \Downarrow \sigma\ \textrm{and}\ \langle \lstiC{(-P;P)}, \sigma \rangle \Downarrow \sigma
\enspace .
\end{align*}
\end{corollary}
\noindent
Finally, 
let $\llparenthesis \lstiC{P} \rrparenthesis$ be the function from $\Sigma$ to $\Sigma$ that \rsemantics associates with a term \lstiC{P}, that is $\llparenthesis \lstiC{P} \rrparenthesis(\sigma) = \tau$ if and only if $\langle \lstiC{P}, \sigma \rangle \Downarrow \tau $, according to \rsemantics as given in Definition~\ref{definition:rsemantics of SCORE}.
Corollary~\ref{corollary:score_is_reversible} implies:
\begin{corollary}
\label{corollary:SCORE term are bijections}
For every term \lstiC{P} of \SCORE, the function
$\llparenthesis \lstiC{P} \rrparenthesis$ is a \emph{bijection}.
\end{corollary}

\subsection{\rsemantics manages failures of \asemantics reversibly}
\label{subsection:reversibleFailureManagement}

Let us focus on the two following observations:
\begin{enumerate} 
\item 
Interpreting a \SCORE term \lstiC{P} according to \rsemantics from state $\sigma$ cannot abort even \MH{when}
a variable \lstiC{x} would be broken due to attempts to apply \vCoq{pop} in a non reversible context, e.g. when trying to \vCoq{pop} non-existing elements from its stack. Instead, the execution of the \rsemantics continues until it reaches a terminal state.

\item 
Conversely, interpreting \lstiC{P} according to \asemantics from a state $\sigma$ results in an immediate abort as soon as \lstiC{x} becomes broken in some state $\tau$, as the computation leading to $\tau$ cannot be ``undone'' uniquely, the term \lstiC{P} \emph{fails}.
\end{enumerate} 
\noindent 
All this means that \asemantics \emph{fails} in interpreting \lstiC{P} if and only if \rsemantics halts on \lstiC{P} in a state where at least one variable remains \emph{broken}. This allows us to conclude that \rsemantics recovers all the situations in which \asemantics would fail.

%% file: S-core-definitive-sos.tex
\begin{figure}[t]
	\centering
	\begin{tabular}{c}
		\bottomAlignProof
		\AxiomC{$ $}
		\RightLabel{\scriptsize \textsc{skip}}
		\UnaryInfC{$\langle \lstiC{SKIP}, \sigma \rangle \Downarrow \sigma$}
		\DisplayProof
		\qquad
		\bottomAlignProof
		\AxiomC{$\sigma(\lstiC{x}) = (v, s, c)$}
		\RightLabel{\scriptsize \textsc{inc}}
		\UnaryInfC{$\langle \lstiC{INC x}, \sigma \rangle \Downarrow \sigma[\lstiC{x} \to (v+1,s,c)]$}
		\DisplayProof
		\\ \\
		\bottomAlignProof
		\AxiomC{$\sigma(\lstiC{x}) = (v, s, c)$}
		\RightLabel{\scriptsize \textsc{dec}}
		\UnaryInfC{$\langle \lstiC{DEC x}, \sigma \rangle \Downarrow \sigma[\lstiC{x} \to (v-1,s, c)]$}
		\DisplayProof
		\qquad
		\bottomAlignProof
		\AxiomC{$\sigma(\lstiC{x}) = (v, s, c)$}
		\RightLabel{\scriptsize \textsc{push}}
		\UnaryInfC{$\langle \lstiC{PUSH x}, \sigma \rangle \Downarrow \sigma[x \to \mbox{\lstinline[language=Coq]|push|}(v,s,c)]$}
		\DisplayProof
		\\ \\
		\bottomAlignProof
		\AxiomC{$\sigma(\lstiC{x}) = (v, s, c)$}
		\RightLabel{\scriptsize \textsc{pop}}
		\UnaryInfC{$\langle \lstiC{POP x}, \sigma \rangle \Downarrow \sigma[x \to \mbox{\lstinline[language=Coq]|pop|}(v,s,c)]$}
		\DisplayProof
		\qquad
		\bottomAlignProof
		\AxiomC{$\langle \lstiC{P} , \sigma \rangle \Downarrow \nu $}
		\AxiomC{$\langle \lstiC{Q} , \nu \rangle \Downarrow \tau $}
		\RightLabel{\scriptsize \textsc{seq}}
		\BinaryInfC{$\langle \lstiC{P;Q}, \sigma \rangle \Downarrow \tau$}
		\DisplayProof
		\\ \\
		\bottomAlignProof
		\AxiomC{$\sigma(\lstiC{x}) = (v,s,c)$}
		\AxiomC{$v \geq 0$}
		\AxiomC{$\langle \lstiC{P}, \sigma \rangle \Downarrow^v \tau$}
		\RightLabel{\scriptsize \textsc{for-fwd}}
		\TrinaryInfC{$\langle \lstiC{FOR x \{P\}}, \sigma \rangle \Downarrow \tau$}
		\DisplayProof
		\\ \\
		\bottomAlignProof
		\AxiomC{$\sigma(\lstiC{x}) = (v,s,c)$}
		\AxiomC{$v < 0$}
		\AxiomC{$\invSC{P} = \lstiC{Q}$}
		\AxiomC{$
         \langle \lstiC{Q}, \sigma \rangle \Downarrow^{-v} \tau$}
		\RightLabel{\scriptsize \textsc{for-bwd}}
		\QuaternaryInfC{$\langle \lstiC{FOR x \{P\}}, \sigma \rangle \Downarrow \tau$}
		\DisplayProof
		\\ \\
		\bottomAlignProof
		\AxiomC{$ $}
		\RightLabel{\scriptsize \textsc{base}}
		\UnaryInfC{$\langle \lstiC{P}, \sigma \rangle \Downarrow^0 \sigma$}
		\DisplayProof
		\qquad
		\bottomAlignProof
		\AxiomC{$\langle \lstiC{P}, \sigma \rangle \Downarrow \nu $}
		\AxiomC{$\langle \lstiC{P}, \nu \rangle \Downarrow^{n} \tau$}
		\RightLabel{\scriptsize \textsc{for-rec}}
		\BinaryInfC{$\langle \lstiC{P}, \sigma \rangle \Downarrow^{n+1} \tau$}
		\DisplayProof
	\end{tabular}
	\caption{\rsemantics of \SCORE.}
	\label{fig:scoreSemanticDefinitive}
\end{figure}

%% file: conclusion.tex
\section{Conclusion}
\SCORE and \rsemantics provide a working example in which reversible programs can \emph{always} be interpreted as \emph{total bijections}. Specifically, the use of \lstiC{assert} is not required to formalize a reversible operational semantics, provided that the state structure on which the semantics operates is sufficiently refined. 

This refinement is precisely what we achieve by transitioning from \asemantics to \rsemantics -- the former implementing \PIF and the latter \TIF, terminology that we introduced in Section~\ref{section:Introduction}. Unlike \asemantics, which fails in certain states, \rsemantics continues computation, \emph{effectively managing failures reversibly}. 

Based on \SCORE and \rsemantics, we aim at determining whether a general approach exists for identifying the necessary structural updates in state representation to transition from \PIF to \TIF.